\documentclass[10pt,a4paper,UKenglish,autoref,thm-restate]{article}
\usepackage{booktabs}
\usepackage{xspace}
\usepackage[margin= 1in]{geometry}
\usepackage{hyperref}
\usepackage{graphicx}
\usepackage{subcaption}
\usepackage[T1]{fontenc}
\usepackage[utf8]{inputenc}
\usepackage{xcolor}
\usepackage{amsfonts,amsmath,amsthm,amssymb}
\usepackage{cleveref}
\newtheorem{theorem}{Theorem}

\newtheorem{lemma}[theorem]{Lemma}
\newtheorem{proposition}[theorem]{Proposition}
\newtheorem{claim}[theorem]{Claim}
\newtheorem{claim*}[theorem]{Claim}

\definecolor{defblueee}{rgb}{0.1,0.4,0.6} 

\renewcommand{\emph}[1]{{\textcolor{red!50!black}{\em #1}}}
\newcommand{\R}{\ensuremath{\mathbb{R}}\xspace}
\newcommand{\mkmcal}[1]{\ensuremath{\mathcal{#1}}\xspace}

\newcommand{\D}{\mkmcal{D}}

\newcommand{\RR}{\mkmcal{R}}
\newcommand{\radius}{r}
\newcommand\dist{\ensuremath{\mathrm{dist}}}

\newcommand{\eps}{\ensuremath{\varepsilon}\xspace}

\title{Near-Linear and Parameterized Approximations for Maximum Cliques in Disk Graphs}

\author{Jie Gao\thanks{Computer Science Department, Rutgers University, Email: \texttt{jg1555@rutgers.edu}}\and Pawel Gawrychowski \thanks{
Institute of Computer Science, University of Wroc\l{}aw, Email: \texttt{gawry@cs.uni.wroc.pl}}\and Panos Giannopoulos\thanks{Department of Computer Science, City St George's, University of London, Email: \texttt{panos.giannopoulos@city.ac.uk} }\and
Wolfgang Mulzer\thanks{Institut f\"ur Informatik, Freie Universit\"at Berlin, Email: \texttt{mulzer@inf.fu-berlin.de}} \and Satyam Singh \thanks{Department of Computer Science, Aalto University, Email: \texttt{satyam.singh@aalto.fi}} \and
Frank Staals \thanks{Department of Information and Computing Sciences, Utrecht University, Email: \texttt{f.staals@uu.nl}}
\and Meirav Zehavi \thanks{Institute for the Theory of Computing, Ben Gurion University of the Negev, Email: \texttt{zehavimeirav@gmail.com} }}

\begin{document}

\maketitle

\begin{abstract}
A \emph{disk graph} is the intersection graph of (closed) disks in the plane. 
We consider the classic problem of finding a maximum clique in a disk graph. 
For general disk graphs, the complexity of this problem is still open, but for unit disk graphs, 
it is well known to be in P. The currently fastest algorithm runs in 
time $O(n^{7/3+ o(1)})$, where $n$ denotes the number of 
disks~\cite{EspenantKM23, keil_et_al:LIPIcs.SoCG.2025.63}. 
Moreover, for the case of disk graphs with $t$ distinct radii, the problem has also recently 
been shown to be in XP. More specifically, it is solvable in time 
$O^*(n^{2t})$~\cite{keil_et_al:LIPIcs.SoCG.2025.63}.
In this paper, we present algorithms with improved running times by allowing for 
approximate solutions and by using randomization: 

\begin{enumerate}
\item[(i)] for unit disk graphs, we give an algorithm that, with constant success probability, computes a $(1-\eps)$-approximate maximum clique in expected time $\tilde{O}(n/\eps^2)$; and

\item[(ii)] for disk graphs with $t$ distinct radii, we give a parameterized approximation 
scheme that, with a constant success probability, computes  a $(1-\eps)$-approximate maximum clique in expected time
$\tilde{O}(f(t)\cdot (1/\eps)^{O(t)} \cdot n)$, for some (exponential) function $f(t)$.
\end{enumerate}

\end{abstract}

\section{Introduction}
\label{sec:Introduction}

A \emph{disk graph} is the intersection graph of closed disks in the plane, where the vertices are 
the disks, and two vertices are connected by an edge if and only if the two corresponding disks intersect.
A \emph{unit disk graph} is a disk graph in which all disks have the same radius, which can be assumed to
be $1$. Disk graphs and unit disk graphs are probably the most extensively studied classes of geometric
intersection 
graphs~\cite{Bonamy2021-ep,bonnet_et_al:LIPIcs.FSTTCS.2020.17,BreuK98,ClarkCJ90,Fishkin03,TkachenkoW25}. 
In particular, they are popular models of wireless communication networks, where each disk represents a 
wireless station, and the respective radii encode the transmission ranges. Theoretically,  
unit disk graphs and disk graphs are of special interest, because they generalize the familiar
class of planar graphs. Indeed, the well-known circle-packing theorem (also called the the 
Koebe–Andreev–Thurston theorem~\cite{Koebe1936}) shows that every planar graph is a disk graph. 
Unlike planar graphs, disk graphs (or unit disk graphs) can be dense, with potentially $\Theta(n^2)$ 
edges. Thus, algorithms on disk graphs often use an implicit representation 
(of the coordinates and radii of the disks) and derive edges on-demand.

Due to the special structure of disk graphs, it is an active research direction to study efficient 
approximation algorithms for classical NP-hard problems on disk graphs or unit disk graphs, such as maximum independent set~\cite{Chan2012-pk,Hunt1994-tu} and minimum dominating 
set~\cite{Gibson10algorithms,Nieberg2008-zi}. In this paper, we study the maximum clique problem on unit 
disk graphs and disk graphs. For a given graph $G=(V, E)$, the problem asks for a vertex set $S\subseteq V$ of maximum cardinality such that the induced subgraph on $S$ is a clique. 
The maximum  clique problem on general graphs is a classical NP-hard problem 
and is also hard to approximate~\cite{Feige2002-va}. However, maximum clique can be solved in polynomial time on unit disk graphs and the problem has been actively studied recently. We first review prior work on this problem and then present our new results. 

\subparagraph*{Related work.} 
Clark, Colbourn, and Johnson~\cite{ClarkCJ90} gave an elegant
$O(n^{4.5})$-time algorithm for finding a maximum clique in a unit
disk graph. The algorithm guesses in quadratic time the pair of most
distant (diameter) disk centers in a maximum clique and considers the
subgraph induced by all the disks whose centers are at most as distant
from both candidate diameter centers. The induced graph is
co-bipartite, hence computing a maximum clique is equivalent to
finding a maximum independent set in the complement bipartite
graph. This in turn reduces to finding a maximum matching. By reducing the
necessary number of pairs of centers to search for and/or reducing the running time for finding a maximum independent set, the running time has been successively improved to $O(n^{3.5} \log n)$ by Breu~\cite{Breu96}, $O(n^3\log n)$ by Eppstein~\cite{Eppstein2010-vb}, $O(n^{2.5}\log n)$ by Espenant, Keil and Mondal~\cite{EspenantKM23} and, most recently, to $O(n^{7/3+ o(1)})$ by Chan~\cite{Chan23} as noted in Keil and Mondal~\cite{keil_et_al:LIPIcs.SoCG.2025.63}. In particular, the latter result is based on the clever divide and conquer approach by Espenant, Keil and Mondal~\cite{EspenantKM23}, which considers only a linear number of center pairs, and on the observation that, for each such pair, one can use a small bi-clique cover of size\footnote{The standard $\tilde{O}(\cdot)$ notation hides polylog factors in $n$.} $\tilde{O}(n^{4/3})$ of the resulting bipartite graph to compute a maximum matching in time almost linear to the size of cover. Very recently, Tkachenko and Wang~\cite{TkachenkoW25} gave an algorithm that runs in time $O(n\log n + nK^{4/3+o(1)})$, where $K$ is the size of the maximum clique; this is an improvement to the above bound by Chan~\cite{Chan23} when $K=o(n)$.



For general disk graphs, the complexity of computing the maximum clique has been a long-standing open problem~\cite{AmbuhlW05,EspenantKM23,Fishkin03}. 
Recently, Keil and Mondal~\cite{keil_et_al:LIPIcs.SoCG.2025.63} considered the case where there are $t$ distinct disk radii and gave an algorithm that runs in polynomial time for every constant $t$. In particular, the running time is $O(2^t n^{2t} (f(n)+n^2))$, where $f(n)$ is the time to compute a maximum matching in an $n$-vertex bipartite graph. 
A maximum matching in a bipartite graph can be computed in time $\tilde{O}(m)$, where $m$ is the number of edges in the graph, using the near-linear time max-flow algorithm in the recent breakthrough in~\cite{Chen2025} with standard reductions to bipartite matching, e.g., the $O(m\log m)$-time fractional flow rounding technique in~\cite{Kang2015}; alternatively, it can be computed in time $\tilde{O}(n^2)$ by a combinatorial algorithm~\cite{Chuzhoy2024-lj}.
Further, the problem admits a randomized EPTAS\footnote{EPTAS stands for efficient polynomial-time approximation scheme.} 
with a running time of $2^{\tilde{O}(1/\eps^3)}n^{O(1)}$ 
and an exact sub-exponential time algorithm~\cite{Bonamy2021-ep}.

Finally, the problem has also been studied for intersection graphs of
other classes of objects in the plane and has been shown to be NP-hard for rays~\cite{Cabello2013-yj}, strings~\cite{Keil2022-op}, 
ellipses and triangles~\cite{AmbuhlW05}, and for sets of both rectangles and unit disks~\cite{bonnet_et_al:LIPIcs.FSTTCS.2020.17}.

\subparagraph*{Our contribution.}
The above recent exciting results and the absence of any non-trivial lower bounds for the problem leave some natural open questions. 

For improving the currently best $O(n^{7/3+ o(1)})$-time bound for unit disk graphs, any algorithm that follows the classical line of approach by Clark, Colbourn, and Johnson~\cite{ClarkCJ90} would effectively need to reduce the number of candidate disk centers to $o(n)$ or compute a maximum matching in the resulting complement bipartite graph in $o(n^{4/3})$ time.
Both of these options seem challenging. On the other hand, the algorithm by Tkachenko and Wang~\cite{TkachenkoW25} does not bring an improvement for instances with $\Omega(n)$-size cliques, and dense instances are in fact the hardest for the problem, see for example~\cite{Bonamy2021-ep}. 

In this paper, we show that further improvement is possible, albeit by allowing for $(1-\eps)$-approximate solutions and by using randomization. 
Our first result is as follows:

\begin{enumerate}
\item[(i)] For unit disk graphs, there is an algorithm that computes with constant success probability a $(1-\eps)$-approximate maximum clique in expected time $\tilde{O}(n/\eps^2)$; this result is presented in Section~\ref{sec:Unit_Disk_Graphs}.
\end{enumerate}

In particular, we first random sample two disk centers
inside a constant-size neighborhood of a properly defined grid-cell. We show that these centers
can help us identify a small region that contains (the centers of) a
$(1-\eps)$-approximate maximum clique of the neighborhood, with $\Omega(\eps)$
probability. Moreover, the disks with centers in that small region induce a co-bipartite graph. To save time, our second step, which is given in Section~\ref{sec:Computing_an_Approximate_Clique}, is to actually compute a $(1-\eps)$-approximate maximum clique in this co-bipartite graph in time that is almost linear in the number of disks of the graph. 

Our second result is for the general disk graph setting. As mentioned above, the EPTAS in~\cite{Bonamy2021-ep} runs in $2^{\tilde{O}(1/\eps^3)}\cdot O(n^4)$ time while, for disk graphs with $t$ distinct radii, the algorithm by Keil and Mondal~\cite{keil_et_al:LIPIcs.SoCG.2025.63} runs in $O^*(n^{2t})$ time\footnote{The standard notation $O^*(\cdot)$ hides polynomial factors.}.
This raises the question of whether such exponential dependencies w.r.t. $t$ (on $n$) and $\eps$ can be improved or even removed altogether. We show that this is possible by considering parameterized $(1-\eps)$-approximations w.r.t. $t$ and $\eps$ and, as in the case of unit disk graphs, by using randomization.

\begin{enumerate}
\item[(ii)] For disk graphs with $t$ distinct radii, there is an efficient parameterized approximation scheme (EPAS)\footnote{An efficient parameterized approximation scheme (EPAS), for some parameter $k$, is a $(1-\eps)$-approximation algorithm that runs in $f(k, \eps) n^{O(1)}$ expected time~\cite{FeldmannKEM20}.
} 
that computes with a constant success probability a $(1-\eps)$-approximate maximum clique in $\tilde{O}(f(t)\cdot (1/\eps)^{O(t)} \cdot n)$ time, for some (exponential) function $f(t)$; this result is presented in Section~\ref{sec:Different_Radii}.
\end{enumerate}

Our algorithm is based on the main idea behind the exact algorithm by Keil and
Mondal. Suppose that for each of the radii one can
guess the leftmost and rightmost disks in an optimal
solution. Then, one can find two sets of disks $L$ and $R$ such that $L\cup R$ induces a co-bipartite graph and contains a maximum clique. The algorithm by Keil and
Mondal~\cite{keil_et_al:LIPIcs.SoCG.2025.63} enumerates all possible
choices to find the leftmost and rightmost disks for each radius. To
speed up the running time for $(1-\eps)$-approximation, we bring in three more ideas. First, we can assume w.l.o.g that for disks of radius $r_i$ that appear in the optimal clique, at least $\Theta(\eps k^*/t)$ of them appear in the max clique where $k^*$ is the optimal clique size, which can be approximated by a factor of $1/5$ in $O(n\log n)$ time  (see~\cref{sub:The_complete_algorithm}) -- otherwise, we can skip disks of radius $r_i$ completely without losing more than $\eps$ fraction. Second, for each radius $r_i$ we can randomly sample two disks $a_i$ and $b_i$ -- if we repeat this enough times, one such pair is close to, in terms of rank, the leftmost and rightmost disks of radius $r_i$ in the optimal solution.  
Last, as in the case of unit disks, we compute a $(1-\eps)$-approximate clique in the resulting co-bipartite graph using the algorithm from Section~\ref{sec:Computing_an_Approximate_Clique}. 



We note here that our result for unit disk graphs is not subsumed by the one for $t$ distinct radii since the latter gives a worse dependency on $\eps$ for $t=1$. 



\subsection{Definitions and notation} 
Let \D be a set of $n$ disks in the plane, and let $G(\D)$ be the \emph{intersection graph} or 
\emph{disk graph} of $\D$ where 

$$G(\D) = (\D, \{  \{D,D'\} \mid D, D' \in \D, D \neq D', \,\, \text{and} \,\, D\cap D' \neq \emptyset \}).$$


A \emph{unit disk graph} is a disk graph where all disks have unit radius, which we will assume to be $1$.
We will focus on unit disk graphs and on disk graphs where the disks have $t$ different
radii, for some given $t>1$. 

Let $D(c,r)$ be a disk with center $c$ and radius $r$. If the radius is clear from the context, we may simply write $D(c)$.
We refer to the intersection of two disks $D(c,r) \cap D(c',r')$ as
their \emph{lens}. 
The line through $c$ and $c'$ (oriented towards
$c'$) splits this lens into two \emph{half-lenses}; a region left of
this line and a region right of this line. See \Cref{fig:lens}.
\begin{figure}[htbp]
    \centering
    \includegraphics{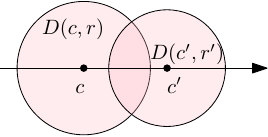}
    \caption{The lens $D(c,r) \cap D(c',r')$ of two intersecting disks $D(c,r), D(c',r')$.}
    \label{fig:lens}
\end{figure}

Let $a$ and $b$ be two points in the plane. The Euclidean distance between $a$ and $b$ is denoted by $\dist(a,b)$. The \emph{vertical slab} of $a$ and $b$ is the region between the vertical lines through $a$ and $b$, that is, the set of all points whose $x$-coordinates lie between the $x$-coordinates of $a$ and $b$.

A graph is \emph{co-bipartite} if its complement is bipartite. Let $X$ and $Y$ be sets of disks such that $G(X)$ and $G(Y)$ are both cliques. As observed in~\cite{ClarkCJ90}, $G(X \cup Y)$ is co-bipartite.
Let $A$ be a set. We denote its cardinality by $|A|$.
Let $X^* \subseteq \D$ be a maximum-size clique in $G(\D)$, i.e., a subset of disks that all pairwise intersect. A clique $X \subseteq \D$ of size at least $(1-\eps)|X^*|$ is a 
$(1-\eps)$-approximate (maximum) clique.

\section{Approximating the Maximum Clique in a co-Bipartite Disk Graph}
\label{sec:Computing_an_Approximate_Clique}

In this section, we give a near-linear-time algorithm for computing a $(1-\eps)$-approximate clique in a co-bipartite disk graph
$G :=G(X \cup Y)$. In such a graph, every $x,x'\in X$ and every $y, y' \in Y$ are adjacent, 
while $x\in X$ and $y\in Y$ are adjacent if and only if the distance between the corresponding centres is at most the
sum of the corresponding radii. We begin with the straightforward transformation: a clique in $G$ corresponds to
an independent set in the complement of $\overline{G}$, which is bipartite. Thus, it is enough to compute
a $(1-\eps)$-approximate independent set in a bipartite graph $\overline{G}$. In this graph, we have an
edge $(x,y)$, for $x\in X$ and $y\in Y$, if and only if $\dist(x, y) > \radius(x)+\radius(y)$, where $\radius(\cdot)$
denotes the radius and $\dist(\cdot,\cdot)$ is the distance between the centres.

We first focus on computing the \emph{cardinality} of a $(1-\eps)$-approximate independent set in $\overline{G}$. 
Recovering the corresponding independent set is a bit more involved and will be explained later.
For any graph, the complement of an independent set is a vertex cover. Denoting by $n$ the number of nodes
in $\overline{G}$, this immediately gives us that the cardinality of the maximum independent set $I^{*}$
is equal to $n$ minus the cardinality of the minimum vertex cover $C^{*}$. We first observe that
a similar property holds for their approximations.

\begin{proposition}
\label{lem:approx}
Let $c$ be an $(1+\eps)$-approximation of $|C^{*}|$, that is, $|C^{*}| \leq c \leq (1+\eps)|C^{*}|$.
Then, $ n-c$ is an $(1-\eps)$-approximation of $|I^{*}|$, that is, $(1-\eps)|I^{*}| \leq n-c \leq |I^{*}|$.
\end{proposition}

\begin{proof}
Because $|C^{*}|=n-|I^{*}|$, we have  $n-c \leq |I^{*}|$.
We observe that $|I^{*}| \geq n/2$, as one of the sides of $\overline{G}$ consists of at least $n/2$ nodes
forming an independent set.
Then, $n-c \geq n-(1+\eps)|C^{*}| = n-(1+\eps)(n-|I^{*}|) = |I^{*}|-\eps \cdot  n+\eps |I^{*}| \geq (1-\eps)|I^{*}|$.
\end{proof}

We now focus on obtaining a $(1+\eps)$-approximation of $|C^{*}|$. For a bipartite graph, Kőnig's theorem states
that the cardinality of the minimum vertex cover is equal to the cardinality of the maximum matching. Thus, denoting
by $M^*$ the maximum matching, it is enough to compute a $(1-\eps/2)$-approximation of $|M^{*}|$, denoted $m$.
Indeed, we can assume $\eps \leq 1$, and then $(1-\eps/2)|M^{*}| \leq m \leq |M^{*}|$ implies
$|M^{*}| \leq m/(1-\eps/2) \leq |M^{*}|/(1-\eps/2) \leq (1+\eps)|M^{*}|$, so $m/(1-\eps/2)$ can be
returned as a $(1+\eps)$-approximation of $|C^{*}|$.

To design the algorithm, we need an efficient data structure for maintaining a dynamic set of points $S\subseteq \mathbb{R}^{2}$.
Each $s\in S$ has its associated weight $w_{s}\in\mathbb{R}$, and we define a distance function $\delta:\mathbb{R}^{2}\times S\rightarrow \mathbb{R}$
as $\delta(p,s)=w_{s}+\dist(p,s)$. The query returns, for a given point $x\in \mathbb{R}^{2}$, the furthest point in $S$,
that is, $\arg\max_{s\in S} \delta(x,s)$. This can be seen as a point location query in the furthest site additively weighted
Voronoi diagram of $S$.

\begin{lemma}
\label{lem:voronoi}
We can maintain a dynamic set of of weighted points $S$ of size $n$ such that an insertion
takes amortized expected time $O(\log^2 n)$, a deletion takes amortized expected time 
$O(\log^4 n)$, and so that a point location query in the furthest site additively weighted Voronoi diagram of $S$ can be answered in worst-case deterministic time $O(\log^2 n)$.
\end{lemma}

\begin{proof}
As explained by Agarwal, Efrat, and Sharir~\cite[Section~9.1]{AgarwalES00}, the
problem reduces to maintaining the upper envelope of a set of
$n$ totally defined continuous bivariate functions
of constant description complexity in three dimensions, under insertions,
deletions, and vertical ray shooting queries. The discussion of
Agarwal, Efrat, and Sharir~\cite[Section~9.1]{AgarwalES00} also shows
that the complexity of these upper envelopes (and thus of the corresponding
farthest-site Voronoi diagrams) is linear.
Thus, the result follows from the data structure given in \cite[Theorem~8.3]{DBLP:journals/dcg/KaplanMRSS20} and its subsequent improvement \cite{Liu22}. 
\end{proof}


With the above data structure in hand, we can describe the algorithm for computing a $(1-\eps)$-approximate matching in $\overline{G}$.
As explained above, returning $n$ minus its cardinality gives us the cardinality of a $(1-\eps)$-approximate
clique in $G$.

\begin{lemma}
\label{lem:matching}
Let $G$ be a co-bipartite disk graph on $n$ vertices. A $(1-\epsilon)$-approximate matching of
$\overline{G}$ can be computed in expected $O((n/\eps)\log^4 n)$ time. 
\end{lemma}

\begin{proof}
Let the given co-bipartite graph be $G(X \cup Y)$. We will call $X$ the 
\emph{left side}, and $Y$ the \emph{right side}.
We recall that the algorithm of Hopcroft and Karp~\cite{HopcroftKarp73} proceeds in phases while
maintaining the current matching $M$ (initially empty). In the beginning of each phase, 
we run a breadth-first search to partition the nodes of the graph into \emph{layers}. 
We start the search from the unmatched nodes on the left side of the graph that form layer $0$.
On even layers of the search, we can traverse any unmatched edge. On
odd layers of the search, we can only traverse a matched edge.
The search terminates after the first layer containing an unmatched node on the right
side of the graph (or when no further layers can be created, meaning that the current 
matching $M$ is a maximum matching). This gives us the length $k$ of a shortest augmenting path.
Next, the algorithm finds a maximal set of node-disjoint augmenting paths of length $k$. This can 
be implemented by running a depth-first search from the unmatched nodes in the last layer. 
The depth-first search is only allowed to follow an edge leading to a yet unvisited node in the
previous layer, and the paths in its tree must alternate between matched and unmatched edges, 
as in the breadth-first search.

The breadth-first search and the depth-first search run in time linear in the size of 
the graph, and the length of a shortest augmenting path increases by at least 2 
in every phase. It is well-known that if the shortest augmenting path is of
length at least $2/\eps$ then the current matching $M$ is a $(1 - \eps)$-approximate 
maximum matching. Thus, it is enough to execute $1/\eps$ phases to obtain such a 
matching. However, the number of edges in our graph might be
even quadratic in $n$, so to finish the proof we need to show how to
implement a single phase in $O(n\log^4 n)$ expected time.

We first explain how to implement the breadth-first search in expected 
$O(n\log^4 n)$ time. To this end, we maintain an instance of the data structure from 
\Cref{lem:voronoi}, initially containing the set $S$ of all points corresponding
to the nodes on the right side of the graph.
We set $w_{s} = -r$ for a point $s\in S$ corresponding to a disk of radius $r$.
Then, on every even layer, we iterate over the nodes there, and for each such node 
$u$ corresponding to a disk $D(c,r)$, we need to iterate over the yet unvisited nodes $v$ 
such that $(u,v)$ is an edge. This is equivalent to obtaining (and removing) all 
points $s\in S$ such that $\delta(c,s) > r$, which in turn can be implemented by repeatedly
retrieving  $\arg\max_{s\in S} \delta(c,s)$, checking if $\delta(c,s) > r$, and if so,
removing $s$ from $S$ and repeating.
On every odd layer, we do not need the data structure, 
as each node has at most one incident matched edge.
The overall number of operations on the data structure is $O(n)$, and the deletions dominate, 
so the total time for the breadth-first search is as claimed.
For the depth-first search, we maintain a separate instance of the data structure from
\Cref{lem:voronoi} for each
even layer $\ell$, initially containing the set $S_{\ell}$ of 
all points corresponding to its nodes. Then, if we are
currently at node $v$ on an odd layer $\ell$, we need
to iterate over the yet unvisited nodes $u$ from layer $\ell-1$ such that $uv$ is an edge.
Let the disk corresponding to $v$ be $D(c,r)$, Then,  this is again equivalent to retrieving 
(and removing) all points $s\in S_{\ell-1}$ such that $\delta(c,s) > r$, and it 
can be implemented as above.
The overall number of operations on all data structures is $O(n)$, making the total time as required.
\end{proof}

For the case of unit disk graphs, \Cref{lem:matching} can be improved. For this, we just
need a  data structure that plays the role of the one in \Cref{lem:voronoi} and that is optimized for unit disk graphs. The following lemma shows how to do this,
adapting a method by Efrat, Itai, and Katz~\cite{Efrat01} to our setting:
\begin{lemma}
\label{lem:vorUDG}
Let $S$ be a set of $n$ unit disks. Then, we can preprocess $S$ in $O(n \log n)$  time
into a data structure that supports the following operations:
\begin{itemize}
    \item $\text{find}(q)$: given $q \in \mathbb{R}^2$, find a disk in the data structure that does
      \emph{not} contain $q$, or report that no such disk exists (i.e.,  $q$ lies
      in the intersection of all the disks in the data structure); and
    \item $\text{delete}(s)$: delete the unit disk $s$ from the data structure.
\end{itemize}
The time for a find-operation is $O(\log n)$, and the total time to delete all
the disks from the data structure is $O(n \log n)$.
\end{lemma}

\begin{proof}[Proof sketch]
Our data structure is very similar to the structure by Efrat, Itai, and 
Katz~\cite[Section~5.1]{Efrat01}, and we refer the reader to their paper
for further details. 

Let $\mathcal{G}$ be a regular grid with cells of edge length $1/2$.
During the preprocessing phase, we locate for every disk center in $S$
the cell of $\mathcal{G}$ that contains it. Let $C$ be a nonemtpy grid cell,
and let $S_C$ be the unit disks whose centers lie in $C$. Since the diameter of
$C$ is less than $1$, the disks in $S_C$ have non-empty intersection (e.g.,
the center of $C$ must lie in all the disks of $S_C$). We compute
the intersection $P_C$ of $S_C$. The intersection $P_C$ is a convex polygon-like structure whose
boundary consists of arcs from $S_C$, where every disk from $S_C$ contributes at most
one arc. The arcs appear along the upper and the lower boundary
of $P_C$ in reverse order of the $x$-coordinates of the corresponding
disk centers. Thus, the total complexity of $P_C$ is  $O(|S_C|)$, and
it can be computed in time $O(|S_C| \log |S_C|)$, using a divide-and-conquer
strategy.
As in the data structure by Efrat, Itai, and Katz, we compute for each $P_C$
two binary trees that represent the structure of the upper and the lower boundary of $P_C$,
respectively, as in the well-known data structure by Overmars
and van Leeuwen~\cite{ovL}. In these trees, the leaves 
correspond to the disks from $S_C$, in reverse $x$-order of their
centers. The inner nodes represent the upper (or the lower)
boundary of the intersection of the disks in the corresponding subtree (see
the paper by Efrat, Itai, and Katz for the details). The binary trees can be
computed in $O(n \log n)$ time, and the total time for the
preprocessing phase is thus $O(n \log n)$.

A find-query can be answered as follows: given $q \in \mathbb{R}^2$, we first check if
there is a non-empty grid cell $C$ that is far enough from $q$ such that $q$ must
lie outside of $P_C$. If such a cell $C$ exists, we report an arbitrary disk
from $P_C$ and are done (this takes $O(\log n)$ time, assuming that the nonempty
grid cells are stored in a suitable data structure). Otherwise, there is a 
constant number of remaining candidate cells, and for each such candidate
cell $C$, we determine whether $q$ lies outside of $P_C$. This can be done
using the  tree structures for $C$ to determine whether $q$ lies
above or below the upper and the lower boundary of $P_C$. If $q$ lies outside
of $P_C$, we can use the result of the tree-searches to 
obtain a disk in $S_C$ that does not contain $q$, as desired.
Again, this takes $O(\log n)$ time.

To delete a a disk $s$ from $S$, we find the cell $C$ that contains the center
of $s$, and we delete $s$ from $S_C$. For this, we must update the tree
structures that represent the intersection $P_C$. The details for the
algorithm and its analysis are given by Efrat, Itai, and Katz, resulting 
in the
claimed running time.
\end{proof}

Using \Cref{lem:vorUDG}, we get the following improved version of \Cref{lem:matching}
for unit disk graphs:

\begin{lemma}
\label{lem:matchingUDG}
Let $G$ be a co-bipartite unit disk graph on $n$ vertices. A $(1-\epsilon)$-approximate matching of
$\overline{G}$ can be computed in expected $O((n/\eps)\log n)$ time. 
\end{lemma}
\begin{proof}
The proof is identical to the proof of \Cref{lem:matching}, but using the
data structure from \Cref{lem:vorUDG} instead of \Cref{lem:voronoi}.
\end{proof}

 \Cref{lem:matching} and \Cref{lem:matchingUDG} allow us to compute the cardinality of the sought $(1-\eps)$-approximate clique.
Computing the corresponding subset of nodes is a bit more involved.

\begin{theorem}
\label{lem:approximate_clique}
Let $G$ be a co-bipartite disk graph with $n$ vertices. We can compute a $(1-\eps)$-approximate clique of $G$ in expected $O((n/\eps)\log^4 n)$ time. When $G$ is a unit disk graph, such a clique can be computed in expected $O((n/\eps)\log n)$ time.
\end{theorem}

\begin{proof}
We provide the analysis for general disk graphs. The case of unit disk
graphs is handled in the same way, using the faster data structure
from \Cref{lem:vorUDG}.

Our strategy is to find a $(1+\eps)$-approximate vertex cover in $\overline{G}$  
in $O((n/\eps)\log^4 n)$ expected time, and we return its complement. By 
\cref{lem:approx}, this constitutes a $(1-\eps)$-approximate clique in $G$.

To obtain the approximate vertex-cover, we proceed as follows:  recall that $X$ and $Y$ are the nodes
on the left and on the right side of $G$.
As in the proof of \cref{lem:matching}, we start with executing $\ell=1/\eps$ phases of the 
Hopcroft-Karp algorithm. This takes expected time $O((n/\eps)\log^4 n)$, and it produces a
$(1-\eps)$-approximate maximum matching $M$. Next, we again run a breadth-first search to 
partition the nodes into layers. Now, however, we do not need to
terminate the search after reaching the first layer containing an unmatched node 
of $Y$. The BFS gives us a partition of $X$ and $Y$ into subsets
\[ 
X_{0}, Y_{0}, X_{1}, Y_{1}, X_{2}, Y_{2}, \dots 
\]
corresponding to the consecutive layers, where $X_{0},X_{1},\dots \subseteq X$, and 
$Y_{0},Y_{1},\dots\subseteq Y$,
as well as the remaining part $X' = X\setminus \bigcup_{i}X_{i}$ and 
$Y' = Y\setminus \bigcup_{i} Y_{i}$.
Because we have run $\ell$ phases of the Hopcroft-Karp algorithm, 
the layers $Y_{0},Y_{1},\dots,Y_{\ell-1}$ do not contain any unmatched nodes, and hence 
$\sum_{i}|Y_{i}| = |M|$. We choose $i\in\{0,1,\dots,\ell-1\}$ such that
$|Y_{i}| \leq |M| / \ell$, and  we define the following set $K$:
\[ 
K := (X \setminus (X_{0}\cup X_{1}\cup \dots \cup X_{i})) \cup (Y_{0} \cup Y_{1} \cup \dots \cup Y_{i}) .
\] 
We will now establish that $K$ is a vertex cover for $G$, with size 
$(1 + \eps)|M|$.

Consider an edge $e = xy$, where $x\in X$ and $y\in Y$. We need to show 
that $x\in K$ or $y\in K$. If $x \in  K$, we are done. Thus, suppose that
$x \not\in K$, and we need to establish that $y \in K$. Since $x \not \in K$,
we have $x\in X_{j}$, for some $j \in \{0,1,\dots,i\}$. 
We consider two cases.
\begin{itemize}
\item \textbf{Case 1: $e\in M$}: 
Because $x\in X_{j}$ is matched, we must have $j\geq 1$ and we have reached $x$ through an
augmenting path ending with $e$. Thus, $y\in Y_{j-1} \subseteq K$.
\item \textbf{Case 2: $e\not\in M$}: 
We can extend an augmenting path ending at $x$ with $e$ to obtain an augmenting path
ending at $y$. Thus, $y\in Y_{0}\cup Y_{1}\cup\dots\cup Y_{j} \subseteq K$.
\end{itemize}
This shows that $K$ is a vertex cover.

We now move to analyzing the size of $K$. First, we argue that every node 
in $K$ is an endpoint of an edge in $M$.
We prove this separately for the nodes in $X$ and in $Y$.
\begin{itemize}
\item \textbf{Case 1: $x\notin X_{0}\cup X_{1}\cup\dots \cup X_{i}$:} 
Such a node is matched, as it does not belong to $X_{0}$.
\item \textbf{Case 2: $y\in Y_{0} \cup Y_{1} \cup \dots \cup Y_{i}$:} 
None of the sets $Y_{0},Y_{1},\ldots,Y_{\ell-1}$ contain an unmatched node,
so such a node is matched.
\end{itemize}
Next, we consider the situation when an edge $xy \in M$ has two endpoints in $K$. 
Then, $x\notin X_{0}\cup X_{1}\cup\dots \cup X_{i}$
and $y\in Y_{j}$, for some $j\in\{0,1,\dots,i\}$. 
Now, it is possible to extend an augmenting path ending at $y$
by the matched edge $yx$ to obtain an augmenting path ending at $x$, so 
$x\in X_{j+1}$. Thus, $j=i$, and the considered
matched edge connects $x\in X_{i+1}$ and $y\in Y_{i}$. 
By the choice of $i$, there are at most $|M| / \ell$ such edges.
This allows us to bound the number of nodes in $K$ by $|M|+|M|/\ell = (1+\eps)|M|$.

Together, this shows that $K$ is an ($1+\eps)$-approximate vertex cover.

As in Lemma~\ref{lem:matching}, running the breadth-first search takes 
$O(n\log^4 n)$ expected time.
After that, constructing $K$ can easily be done in additional linear time.
\end{proof}

\section{Unit Disk Graphs}
\label{sec:Unit_Disk_Graphs}
We describe an algorithm that computes an $(1-\eps)$-approximate 
clique in a unit disk graph  with probability at least $1 - \delta$ 
in $O((n/\eps^{2})\log(n)\log (1/\delta))$ 
expected time.
Let $P$ be a set of $n$ points in the plane, and let $\mathcal{D}(P) = \{D(p, 1) \mid p \in P\}$ be the set of unit disks whose centers are given by $P$. For convenience, we will sometimes identify disks with their centers. A set $X \subseteq P$ is said to form a clique if all the disks in $\mathcal{D}(X)$ intersect pairwise. 

Let $\mathcal{G}$ be a regular grid where each grid cell has diameter $2$, and let $\mathcal{C}$ be the collection of grid cells that contain at least one point from $P$.
Let $C \in \mathcal{C}$ be such a nonempty cell. 
The \emph{extension} of $C$, denoted by $C^+$, is the $5 \times 5$-block of grid cells whose central cell is $C$. 
Furthermore, let $P_C \subseteq P$ be the set of disk centers that lie in any cell of the extension $C^+$; for an illustration, see~\Cref{fig_extension_C}.
Each point of $P_C$ is uniquely assigned to exactly one cell of $C^+$: if a point lies on the boundary shared by multiple cells of $C^+$, it is assigned arbitrarily to one of them.
Then, a simple diameter argument shows that every clique in $P$ must be contained within the extension of some nonempty cell.
\begin{figure}[ht]
    \centering
    \includegraphics[width=0.38\linewidth]{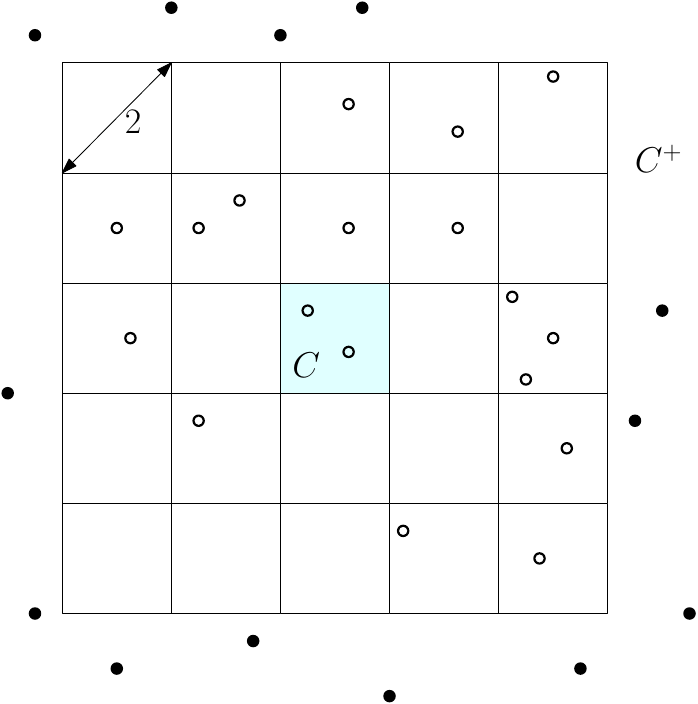}
    \caption{A cell $C$ (the central cell) and its extension $C^+$, where each cell in $C^+$ has a diameter 2. All points (both solid and hollow) belong to $P$, while the hollow points also belong to $P_C$.}
    \label{fig_extension_C}
\end{figure}
\begin{claim}
\label{clm:cliqueext}
Let $X$ be a clique in $P$. Then, there exists a
cell $C \in \mathcal{C}$ such that $X \subseteq P_C$.
\end{claim}
\begin{proof}
Let $p \in X$, and let $C \in \mathcal{C}$ be
the cell that contains $p$. Since every
other disk center in $X$ has distance at most $2$ from $p$,
it follows that $X$ must lie in $C^+$. Hence, we have $X\subseteq P_C$,
as claimed.
\end{proof}

Furthermore, a simple packing argument shows that every extension of a nonempty cell $C$
contains a clique of size linear in $|P_C|$.
\begin{claim}
\label{clm:maxcl}
Let $X^*$ be a maximum clique in $P_C$. Then, 
we have $|X^*| \geq |P_C|/25$.
\end{claim}
\begin{proof}
The extension $C^+$ consists of $25$ grid cells with pairwise disjoint interiors, where each cell has diameter $2$.
Hence, for any such cell $C' \subset C^+$, the set $P_C \cap C'$ must form a clique, since all these disks are pairwise intersecting.
Recall that if a disk center lies on the boundary of a cell, it is assigned arbitrarily to one of the adjacent cells so that each point in $P_C$ is uniquely assigned to exactly one cell in $C^+$.
Since there are $25$ such cells, 
there exists at least one cell in $C^+$ containing at least $|P_C|/25$ disk centers.
It follows that $P_C$ contains a clique of size at least $|P_C|/25$, and therefore the maximum clique $X^*$ is at least of size $|P_C|/25$.
\end{proof}

Now, let $C \in \mathcal{C}$.
We select two disk centers $p_1, p_2 \in P_C$ independently and uniformly at random. Let $s$ be the line segment of length 2 that starts at $p_2$ and is directed toward $p_1$, so that $p_1$ lies either on $s$ or on its extension along the same direction.
Let $x$ be the other endpoint of $s$, and let $L = D(x, 2) \cap D(p_2,2)$ be the lens with axis $s$; see~\cref{fig_lens_L} for an illustration.
We determine the set $P_L = L \cap P_C$ of disks centers that lie in $L$, and we find an $(1-\eps/2)$-approximate maximum clique $X$ for $P_L$, using~\cref{lem:approximate_clique} from~\cref{sec:Computing_an_Approximate_Clique}. 
The lens $L$ is partitioned by segment $s$ into two parts: the part to the left of $s$ and the part to the right of $s$. Without loss of generality, the points on $s$ belong to the left part. Any two points in the same part have a distance at most $2$. Thus, all disks with centers in the same part form a clique. In other words, the intersection graph of unit disks centered at $P_L$ is a co-bipartite graph. %
Consequently, as we show below, the clique $X$ is an
$(1-\eps)$-approximate clique for $P_C$ with
probability $\Omega(\eps)$.
\begin{figure}[t]
    \centering
    \includegraphics[page=2,width=0.36\linewidth]{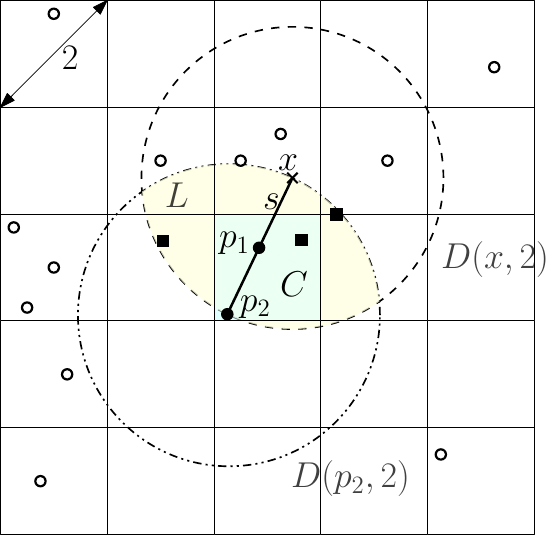}
    \caption{The cell $C$ is represented by the light-blue colored
      cell with $p_1, p_2\in C$. The thick black line denotes the line
      segment $s=\overline{p_2 x}$ of length $2$, passing through
      $p_1$. The dotted disk represents $D(x,2)$, while the
      dash-dotted disk represents $D(p_2,2)$. The yellow-shaded region
      corresponds to the lens $L = D(x,2) \cap D(p_2,2)$, formed along
      the axis $s$.
    }
    
    \label{fig_lens_L}
\end{figure}
\begin{lemma}
\label{lem:gridclique}
Let $X^*$ be a maximum clique in $P_C$.
Then, with probability
$\Omega(\eps)$, we have 
$|X| \geq (1-\eps)|X^*|$.
\end{lemma}

\begin{proof}
First, by~\Cref{clm:maxcl}, we know that $|X^*| \geq |P_C|/25$.
Thus, with probability at least $1/25$, we have that $p_1 \in X^*$.
Now, we sort the centers in $X^*$ in the increasing order of their distances from $p_1$,
and let $\pi^*$ be the resulting ordering of $X^*$. 
Let $X^+$ be those centers from $X^*$  that have rank at least
$(1-\eps/2)|X^*|$ in $\pi^*$, and let $X^- = X^*
  \setminus X^+$ (in other words, $X^+$ contains
the $(\eps/2)|X^*|$ points in $X^*$ that are furthest away from $p_1$,
and $X^-$ contains the $(1-\eps/2)|X^*|$ remaining
points).
Since
\[
|X^+| = \frac{\eps}{2}|X^*| \geq \frac{\eps}{50}\,|P_C|,
\]
it follows that with probability
at least $\eps/50$, the second sampled center $p_2$ lies in $X^+$.

Since $p_1$ and $p_2$ are sampled independently, we
can conclude that  with probability at least
$\eps/1250$,
it holds that $p_1 \in X^*$
and that $p_2 \in X^+$. Henceforth, we assume that indeed this is the case.
Since $p_1$ and $p_2$ both lie in $X^*$, the distance
$\dist(p_1, p_2)$ between $p_1$ and $p_2$ is at most $2$. 
From this, it follows that 
$D(p_1, \dist(p_1, p_2)) \subseteq D(x, 2)$
since by the choice of $x$, we have $\dist(x, p_1) + \dist(p_1, p_2) = 2$, which implies that for every point $y \in  D(p_1, \dist(p_1, p_2))$, the triangle inequality gives
\[
\dist(x, y) \leq \dist(x, p_1) + \dist(p_1, y) \leq \dist(x, p_1) + \dist(p_1, p_2) = 2,
\]
as needed.

As $p_2 \in X^+$, and since all the points in $X^+$ are at least as
far away from $p_1$ as all the points in $X^-$, it follows that 
$X^- \subseteq D(p_1, \dist(p_1, p_2)) \subseteq D(x, 2)$.
Moreover, since $X^- \subseteq X^*$ and $p_2 \in X^*$, we also have
$X^- \subseteq D(p_2, 2)$.
This implies that $X^-$ lies in the lens $L$ defined by
$D(x, 2)$ and $D(p_2, 2)$. In particular, $P_L$ contains a clique of size $(1-\eps/2)|X^*|$, and since $X$ is a $(1-\eps/2)$-approximate clique for $P_L$, we have
\[
|X| \geq  (1-\eps/2)(1-\eps/2) |X^*| = 
(1 - \eps + \eps^2/4)|X^*| \geq
(1 - \eps)|X^*|.
\]
This means that $X$ is an $(1-\eps)$-approximate maximum clique for $P_C$, as claimed.
\end{proof}

We can now present the complete algorithm.


\begin{theorem}
\label{thm:approx_clique_in_udg}
Let $P$ be a set of $n$ points in the plane, let $\eps > 0$ be a
parameter. There is a randomized algorithm that
runs in expected time
$O((n/\eps^2)\log n)$ and computes an
$(1-\eps)$-approximate clique for $\mathcal{D}(P)$ with probability at least $1/2$. 
\end{theorem}

\begin{proof}
Let $C \in \mathcal{C}$ be a nonempty grid
cell. By~\Cref{lem:gridclique}, and using~\Cref{lem:approximate_clique}, one can compute, in $O((|P_C|/\eps) \log n)$ expected time, a clique for $P_C$ that, with probability $\Omega(\eps)$, is an $(1-\eps)$-approximate clique for $P_C$.
By repeating this process $O(1/\eps)$ times, and selecting the largest of
the resulting cliques, one can compute in $O((|P_C|/\eps^2) \log n)$
expected time, a clique that is an $(1-\eps)$-approximate clique for
$P_C$ with probability at least $1/2$.

We do this for $P_C$ of every nonempty cell $C \in \mathcal{C}$, and
we return the largest of the resulting cliques. (Note that all non-empty cells can be easily found in $O(n\log n)$ time~\cite{Wang23}). 
By~\Cref{clm:cliqueext}, there exists a cell $C \in \mathcal{C}$ such that $P_C$ contains a max-clique of $P$, so with probability at least $1/2$, this gives an $(1-\eps)$-approximate maximum clique for $P$.
The total running time is proportional to 
\[
\sum_{C \in \mathcal{C}} \frac{|P_C|}{\eps^2} \log n = 
\frac{\sum_{C \in \mathcal{C}} |P_C|}{\eps^2} \log n  = 
O((n/\eps^2) \log n), 
\]
since every point from $P$ appears in $O(1)$ sets $P_C$.
The theorem follows.
\end{proof}

By a standard argument, for any value $\delta > 0$, we can raise the probability of success to  $1-\delta$ by  independently repeating the algorithm of~\Cref{thm:approx_clique_in_udg} $O(\log(1/\delta))$ times and by taking the best solution. This increases the running time by an $O(\log(1/\delta))$ factor.

\section{Disk graphs with different radii}
\label{sec:Different_Radii}

We first review the algorithm by Keil and Mondal~\cite{keil_et_al:LIPIcs.SoCG.2025.63}. The main idea is simple and elegant, and it goes as follows. Suppose that there are $t$ distinct radii
$r_1,\dots,r_t$, in increasing
order. One can enumerate all radii that appear in a maximum clique and, for each radius $r_i$, enumerate all disk pairs to find the one with centers $a_i, b_i$ that are, respectively, leftmost and rightmost (in the optimum) along the $x$-axis. Let $X$ be the set of these $2t$ disks.
For each pair of centers $a_i, b_i$, consider their vertical slab and find the set $L_i$ of all radius-$r_i$ disks that intersect all disks in $X$ and have their centers within the slab and above the line segment $a_i b_i$. Similarly, find the set $R_i$ of all radius-$r_i$ disks that intersect all disks in $X$ and have their centers in the slab and below the line segment $a_ib_i$. Take $L=L_1\cup L_2 \cup \dots \cup L_t$ and $R=R_1\cup R_2  \cup\dots\cup R_t$. A simple geometric argument shows that the graph induced by $L\cup R$ is co-bipartite, and thus, similar to the case of unit disks, one can find a maximum clique by finding a maximum matching. The algorithm runs in time $O(2^t n^{2t} (f(n)+n^2))$, where $f(n)$ is the time to compute a maximum matching in an $n$-vertex bipartite graph. The factor $n^{2t}$ comes from the enumeration to find the leftmost and rightmost disks in the optimal solution for every radius $r_i$. 

In this section, we give an efficient parameterized approximation scheme w.r.t.\@ the number $t$ of different radii and $\eps$. There are two places where we can speed up the algorithm described above if approximate solutions are allowed. First, instead of finding the leftmost and rightmost disks for each radius by enumeration, we use random sampling to find two centers $a_i, b_i$ that are close (in terms of rank) to the leftmost and rightmost disks in the optimal clique. Second, as in  the case of unit disk graphs, we can replace finding a maximum clique by finding an approximate maximum clique using~\Cref{lem:approximate_clique}.

Let $\D_i$ be the set of disks of radius $r_i$ and let $D^*$ be a maximum clique, with $|\D^*|= k^*$. In the next subsection, we will assume that $\D^*$ contains at least $\eps k$ disks from each $\D_i$, for some parameter $k \in \mathbb{N}$. Later, in Section~\ref{sub:The_complete_algorithm}, we will drop this assumption using brute-force.

\subsection{A special case: Disks of every radius contribute to the maximum clique}
\label{sub:An_algorithm_for_a_fixed_smallest_disk}
In this section, we will compute an $(1-\eps)$-approximate maximum clique that
has at least one disk from each radius.

Our algorithm is simple: we (independently) pick a disk
$D(o,r_1)$ uniformly at random $m_1 = \Theta(n/k\eps)$ times and compute the subset of disks $\D'$ that intersect $D(o,r_1)$, where $r_1$ is the smallest radius.
With each choice of $D(o, r_1)$ and the corresponding $\D'$, we repeat the following procedure $m_2= \Theta((k^*/k\eps^2)^{2t})$ times:

\begin{enumerate}
\item For each radius $r_i$: uniformly at random pick two disks
  $D(a_i,r_i)$ and $D(b_i,r_i)$ from $\D'_i = \D_i \cap \D'$. Let $X$
  denote the set of the selected disks for all $r_i$. Without loss of generality, assume that $a_i$ is to the left of $b_i$. We interpret $a_i$ as the leftmost disk
  of radius $r_i$ in the clique, and $b_i$ as the rightmost disk in
  the clique.

\item Compute the subset of disks $\D'' \subseteq \D'$ that intersect
  all disks in $X$ (by explicitly checking if a disk $D(c,r) \in \D'$
  intersects all disks in $X$).
\item Form two sets of disks $L$ and $R$. For each vertical slab with  $a_i, b_i$ on the boundary, find all the radius-$r_i$ disks $L_i\subseteq \D''$ with center within the slab and above the line segment $a_i b_i$. Similarly, find all the radius-$r_i$ disks $R_i\subseteq \D''$ with centers in the slab and below the line segment $a_ib_i$. Take $L=L_1\cup L_2\cup\dots\cup L_t$ and $R=R_1\cup R_2 \cup\dots \cup R_t$. 
\item Use the algorithm from Theorem~\ref{lem:approximate_clique} to
  compute an $(1-\eps/2)$-approximate clique in $G(L\cup R)$. Note that this is possible, as $G(L\cup R)$ is co-bipartite, see~\Cref{lem:disks_included} below.
\end{enumerate}
\medskip
We finally return the largest clique among all the repeated $m_1m_2$ runs. 

For the  analysis, we start with the following observation.

\begin{lemma}\label{lem:cone}
For a disk $D(o,r_1)$,  let $\D'$ be the set of disks that intersect $D(o, r_1)$, Then, $|\D'|\leq 6k^*$.
\end{lemma}
\begin{proof}
Consider the six cones of angle $\pi/3$ that partition $\R^2$ and have their apex at $o$. Any two disks $D(c_i,r_i)$ and $D(c_j,r_j)$ with the centers in the same cone and such that both
intersect $D(o,r_1)$ must also pairwise intersect. To see that, we
have $\dist(o,c_i)\leq r_1+r_i$ and $\dist(o,c_j)\leq r_1+r_j$. We also have that $\angle c_i
o c_j\leq \pi/3$. Thus, $\dist(c_i,c_j)\leq \max\{\dist(o,c_i), \dist(o,c_j)\}\leq
r_i+r_j$, since $r_1 \leq r_i, r_j$.
Hence, the disks from $\D'$ that are in a single cone form a clique. 
There must be a cone containing at least $|\D'|/6$ centers from $\D'$, and, consequently, we have $k^* \geq |\D'|/6$. 
\end{proof}

The probability that by random sampling one selects a specific disk $D(o,r_1)$ in $\D^*$ is
at least $\eps k/n$. Since we repeat the algorithm $m_1=\Theta(n/k\eps)$ times, with a constant probability $D(o,r_1)\in \D^*$.
In the following, we assume that $D(o,r_1)$
indeed appears in $\D^*$, and then use a similar argument as in the
unit disk case to argue that from each radius $r_i$, we discard at most 
an $\eps/2$-fraction of the disks from $\D^* \cap \D_i$ in our clique.

Let $c_1,\dots,c_h$ be the centers of the disks in $\D^* \cap D_i$,
ordered from left to right, and note that by our assumption on $\D^*$, we
have $h \geq \eps k$. By Lemma~\ref{lem:cone}, the probability that
$a_i \in \{c_1,\dots,c_{(\eps/4) h}\}$ is
\begin{align*}
  \frac{(\eps/4) h}{|D'_i|} \geq  \frac{(\eps/4) h}{|D'|} \geq \frac{\eps h}{24 k^*} \geq  \frac{\eps^2k}{24 k^*}.
\end{align*}

Similarly, the probability that $b_i$ is among the rightmost $(\eps/4) h$ centers of
$\D^* \cap \D_i$ is also at least $\eps^2k/24 k^*$. 

The following technical lemma, stated for our setting, asserts that we include at least $h-2(\eps/4)h=(1-\eps/2)h$
disks from $\D^* \cap \D_i$ in $L \cup R$.


\begin{lemma}
  \label{lem:disks_included}
Assume that $D(o, r_1)$ and $D(a_i, r_i)$, $D(b_i, r_i)$, for all $i$ are from $\D^*$. The disks in $L$ form a clique. Similarly, the disks in $R$ form a clique. Further, for every $i$, all disks in $\D^*$ of radius $r_i$ with centers within the vertical slab within $a_i$ and $b_i$ are included in $L\cup R$. 
\end{lemma}
\begin{proof}
The proof that $L$ and $R$ form a clique  respectively follows directly from~\cite[Lemma 3.1]{keil_et_al:LIPIcs.SoCG.2025.63}. Next, all disks in $D^*$ intersect $D(o,r_1)$, and $D(a_i, r_i), D(b_i, r_i)$, for all $i$. Thus the disks of radius $r_i$ of $D^*$ with the conditions specified in the lemma
meet the requirements for $L_i$ or $R_i$ and thus are included in $L\cup R$.
\end{proof}


Hence, with probability  at lest $(\eps^2k/24 k^*)^2$ we include at least a
$(1-\eps/2)$ fraction of the disks of $\D^* \cap D_i$ in
$L \cup R$. The probability that we include $(1-\eps/2)$ fraction of
all disks (e.g. over all $t$ radii) is thus at least $(\eps^2k/24 k^*)^{2t}$.

We now compute a $(1-\eps/2)$-approximate clique in $G(L\cup R)$, so the
size is at least $(1-\eps/2)(1-\eps/2)k^* \geq (1-\eps)k^*$.

Repeating this
$m_2 = \Theta((k^*/k\eps^2)^{2t})$ times, the algorithm succeeds with a constant
probability. If  $k = \Theta(k^*)$, then $m_2 = O(1/\eps^{4t})$.

Now we can conclude with the full algorithm and its runtime.

\begin{lemma}
  \label{lem:multi_disk_fixed_radii}
  With a constant success probability, we can compute a $(1-\eps)$-approximate maximum clique of $n$ disks of $t$ distinct radii, in which every radius
  appears at least $\eps k$ times, in expected time
\[ O\left(n\log n+\frac{n}{k\eps}(t\log n+k^*)+\frac{nk^*}{k\eps}\left
          (\frac{k^*}{k\eps^2}\right )^{2t}
          \left(t + 
          \frac{\log^4
          k^*}{\eps}
          \right)
          \right)\]
  where $k^*$ is the size of the maximum clique.
\end{lemma}
\begin{proof}
    For each disk $D(o, r_1)$, finding the disks that intersect it can be done easily in $O(n)$ time. But if we repeat this step $m_1=\Theta(n/k\eps)$ times, this will be too costly. Instead, we pre-process the disks into a data structure and answer a 2D circular range query, which returns the points within a query disk. In particular, for each radius $r_i$ we process the centers of the disks in $\D_i$ into a size $O(|\D_i|)$ data structure using the construction in~\cite{Afshani2009-ih}
    such that one can report all centers of the disks in $\D_i$ within distance $r_1+r_i$ from $o$ in $O(\log n+\ell_i)$ time, where $\ell_i$ is the size of the output. The construction time for the data structure for $\D_i$ is $O(|\D_i|\log|\D_i| )$. Thus total construction time for all $t$ different radii is $O(n\log n)$. We answer $tm_1$ queries and the total query time is $O(tm_1 \log n+m_1k^*)$. 
    
    Note that Steps~1--4 of our algorithm operate on $\D'$ for each choice of $D(o,r_1)$ with $|\D'|=O(k^*)$ disks. 
    Thus, the running time for Steps~1--4 is
    $O(2tk^*+f_{\eps/2}(k^*))$, where $f_\eps(n)$ is the time
    for computing a $(1-\eps)$-approximate clique in an $n$ by $n$ co-bipartite
    disk graph. Applying Theorem~\ref{lem:approximate_clique}, we have
    $f_\eps(n)=O((n/\eps)\log^4 n)$ expected time, so we have one iteration of Steps~1--4 to be $O(tk^*+(k^*/\eps) \log^4 k^*)$.
    Summing all the computation costs, we have,
    \[ O\left(n\log n+\frac{n}{k\eps}(t\log n+k^*)+\frac{nk^*}{k\eps}\left
          (\frac{k^*}{k\eps^2}\right )^{2t}
          \left(t + 
          \frac{\log^4
          k^*}{\eps}
          \right)
          \right).\qedhere
    \]
\end{proof}

\subsection{The complete algorithm}
\label{sub:The_complete_algorithm}




\begin{theorem}
\label{thm:approx_clique_in_disk_graph}
  Let \D be a set of $n$ disks, let $t$ be the number of distinct
  radii in \D, and let $\eps \in (0,1)$ be a parameter.
  With a constant success probability, we can compute a
  $(1-\eps)$-approximate maximum clique in $G(\D)$ in expected time $\tilde{O}\left((\frac{t}{\eps})^{O(t)}n\right)$.   
\end{theorem}
\begin{proof}
First, we compute a $1/5$-approximate maximum clique $\hat{D}$. This can be achieved in $O(n\log n)$ time as follows. It is well-known that every set of pairwise intersecting disks can be stabbed by four points~\cite{Danzer86, Stacho84}. This implies that there exists a point that is covered by at least $1/4$ of the disks in a maximum clique\footnote{Carmi, Katz and Morin~\cite{Carmi2023-av} showed that given a set of pairwise intersecting disks one can compute a set of four stabbing points in linear time but it's not obvious how to find a point as above in the same time.}. Thus, a point of maximum depth is covered by at least so many disks as well. Computing a point of maximum depth in an arrangement of disks is $3$SUM-hard~\cite{AronovH08}; however, a point of $(1-\eps)$-approximate maximum depth can be computed, with high probability, in $O((1/\eps)^2 n\log n)$ expected time~\cite{AronovH08}. We therefore run this algorithm for $\eps=1/5$ and get the approximate clique of the desired size. Thus, we take $k=|\hat{D}|$, where $k^*/5 \leq k \leq k^*$.

Consider the optimal solution $D^*$ with $k^*$ disks and assume that all disks of radius $r_i$ are removed from $D^*$ if there are at most $\eps k^*/(2t)$ of them in it. This leaves a clique $\bar D\subseteq D^*$ with $(1-\eps/2)k^* \leq |\bar D|\leq k^*$. Note that $\bar D$ has at least $\eps k^*/2t\geq \eps k/2t$ of disks in radius $r_i$, if $r_i$ appears at all. We now aim to find a $1-\eps/2$ approximation to $\bar D$, which is a $1-\eps$ approximation to $D^*$.  Of course, we do not know in advance what those radii are that appear less than $\eps k^*/(2t)$ in $D^*$, which is some (fixed but unknown) optimal solution, and therefore we will brute-force through all subsets of radii, at the cost of an extra factor of $2^t$ in the running time.

Let $\RR$ denote the ordered set of distinct radii appearing in
\D. For each subset $R$ of radii, we run the algorithm from
Lemma~\ref{lem:multi_disk_fixed_radii} with parameters $\eps/2$ and
$k/t$. 
We output the best solution.
The correctness and the approximation factor follow immediately from Lemma~\ref{lem:multi_disk_fixed_radii} and the fact that we enumerate all radius subsets, one of which matches the set of radii that appear more than $\eps k^*/(2t)$ times in the optimal solution $D^*$. For that particular run, Lemma~\ref{lem:multi_disk_fixed_radii} guarantees that we find a $1-\eps$ approximation to $D^*$ with a constant probability.
    
    The running time comes from repeating Lemma~\ref{lem:multi_disk_fixed_radii} with parameters $\eps/2$ and $k/t$, for $2^t$ times. Notice that the construction of the circular range query data structure is only done once at the initialization step. But all other steps will be repeated $2^t$ times. 
    Further, $k^*/5 \leq k \leq k^*$. 
    The total running time is obtained from applying Lemma~\ref{lem:multi_disk_fixed_radii}.
$$O\left (n\log n+\frac{2^tt^2n\log n}{k^*\eps}+\frac{t2^tn}{\eps}+ \frac{2^{7t}t^{2t+1}n\log^4 k^*}{\eps^{4t+2}}
+ \frac{nt2^{4t+1}}{\eps^{4t+1}}
\right)=\tilde{O}\left(\left(\frac{t}{\eps}\right)^{O(t)}n\right).$$ 
\end{proof}

As in the case of unit disks, we can increase the probability of success to any desired value by independently repeating the algorithm sufficiently many times. 



\section*{Acknowledgement}
The work in this paper was initiated at the Lorentz Center workshop on Fine-Grained \& Parameterized Computational Geometry, held in February 2025. Gao acknowledges support from NSF under grants CNS-2515159, IIS-2229876, DMS-2220271, DMS-2311064, and CCF-2118953. 
Gawrychowski acknowledges support from the Polish National Science Centre under grant 2023/51/B/ST6/01505.
Singh acknowledges support from the Research Council of Finland under grant 363444.
Zehavi acknowledges support from the Israel Science Foundation under grant 1470/24, and from the
European Research Council under grant 101039913 (PARAPATH).

\bibliography{clique}


\end{document}